\documentclass[11pt]{amsart}
\usepackage{amsthm}
\usepackage{amsmath,amstext,amsthm,amsfonts,amssymb,amsthm}
\usepackage{multicol}
\usepackage{latexsym}
\usepackage{color}
\usepackage{graphicx}
\usepackage{epsf}
\usepackage{epsfig}
\usepackage{epic}
\usepackage{graphics}
\usepackage{verbatim}
\usepackage{color}
\usepackage{hyperref}
\newtheorem{theorem}{Theorem}[section]
\newtheorem{lemma}[theorem]{Lemma}
\newtheorem{corollary}[theorem]{Corollary}
\theoremstyle{definition}
\newtheorem{definition}[theorem]{Definition}
\newtheorem{proposition}[theorem]{Proposition}

\theoremstyle{remark}

\numberwithin{equation}{section}

\newcommand{\qu}{\mathbf{q}}

\begin{document}
\title[Finite frames]{Discrete frames on finite dimensional left quaternion Hilbert spaces}
\author{M. Khokulan$^1$, K. Thirulogasanthar$^2$ and S. Srisatkunarajah$^1$}
\address{$^{1}$ Department of Mathematics and Statistics, University of Jaffna, Thirunelveli, Jaffna, Sri Lanka. }
\address{$^{2}$ Department of Computer Science and Software Engineering, Concordia University, 1455 de Maisonneuve Blvd. West, Montreal, Quebec, H3G 1M8, Canada.}
\email{mkhokulan@gmail.com, santhar@gmail.com}
\thanks{This research is part of an M.Phil thesis to be submitted to University of Jaffna }
\date{\today}
\begin{abstract}
An introductory theory of frames on finite dimensional left quaternion Hilbert spaces is demonstrated along the lines of their complex counterpart. 
\end{abstract}
\maketitle
\pagestyle{myheadings}
\section{Introduction}\label{sec_intro}
Duffin and Schaeffer invented frames while working on non-harmonic Fourier series \cite{DU}. Later, despite the fact that many others investigated the frame theory, the mechanism brought by Daubechies et al. gave a strong place to frames in harmonic analysis \cite{D1,D2}.  Coherent states of quantum optics and wavelets are subclasses of continuous frames \cite{Alibk}. In the modern era, frames became standard tool in many areas of engineering and physical problems partly because of their success in digital signal processing, particularly one can name the time-frequency analysis \cite{Ole, GR}. In this note we are primarily interested in frames on finite dimensional quaternion Hilbert spaces. There has been a constant surge in finding  finite tight frames, largely as a result of several important applications such as internet coding, wireless communication, quantum detection theory, and many more \cite{CG, CGV, GO, Ole, TW}. There is no fixed class of frames which is appropriate for all physical problems. Each problem is different and a solution to it demands a specific frame. The current nature of technology advancement constantly faces new problems, therefore, the research to find tools to solve them will continue.

Any vector in a separable Hilbert space can be written in terms of an orthonormal basis of the Hilbert space and the so-written expression is unique. Besides, finding an orthonormal basis for a separable Hilbert space is hard, the uniqueness of the expression of vectors, in terms of an orthonormal basis, is undesirable in applications. In fact, this uniqueness limited flexibility in applications. In order to overcome this strain practitioners looked for a substitute. At this juncture, frames appeared as a replacement to orthonormal bases. In a finite dimensional Hilbert space, usually, a frame contains more vectors than any orthonormal basis and this surplus allowed vectors to have enormously many expressions. This flexibility of frames is the key to their success in applications. How flexible a frame should be varies according to the nature of each problem. Flexibility of frames provides design freedom and which permits one to build frames adequate to a particular problem in a way that is not possible by an orthonormal basis \cite{Alibk, Ole, D1,A1}.\\

We can only define Hilbert spaces over the fields $\mathbb{R}, \mathbb{C}$ and $H$, which are the set of all real numbers, the set of all complex numbers, and the set of all quaternions respectively \cite{Ad}. The fields $\mathbb{R}$ and $\mathbb{C}$ are commutative and associative. Functional analytic properties of real and complex Hilbert spaces well-studied. However, the field of quaternions is non-commutative but associative. Due to this non-commutativity of quaternions, a systematic study of functional analytic properties of quaternionic Hilbert spaces is not concluded yet. Further, as a result of this non-commutativity we have two types of Hilbert spaces, namely left quaternionic Hilbert space and right quaternionic Hilbert space. The functional analytic properties of the underlying Hilbert space are crucial to the study of frames. In this regard, we shall examine some functional analytic properties of finite dimensional quaternion Hilbert spaces as required for the development of the frame theory.

As far as we know a general theory of discrete frames on quaternionic Hilbert spaces is not available in the literature. In this respect we shall construct discrete frames on finite dimensional left quaternionic Hilbert spaces following the lines of \cite{Ole}. In the construction of frames on finite dimensional quaternion Hilbert spaces the effect of non-commutativity of quaternions is adaptable, and therefore most of the results can be manipulated from their complex counterparts. Finally, as a scope of the construction, four dimensional quaternions may provide more feasibility in applications than its two dimensional complex counterpart. 

\section{quaternion Algebra}
In this section we shall define quaternions and some of their properties as needed here. For details one may consult \cite{Ad, TH, ThiAli}.
\subsection{Quaternions}
Let $H$ denote the field of quaternions. Its elements are of the form $\qu=x_0+x_1i+x_2j+x_3k,~$ where $x_0,x_1,x_2$ and $x_3$ are real numbers, and $i,j,k$ are imaginary units such that $i^2=j^2=k^2=-1$, $ij=-ji=k$, $jk=-kj=i$ and $ki=-ik=j$. The quaternionic conjugate of $\qu$ is defined to be $\overline{\qu} = x_0 - x_1i - x_2j - x_3k$.\\
\subsection{Properties of Quaternions:}
The quaternion product allows the following properties.
For $q,r,s\in H,$ we have 
\begin{enumerate}
	\item [(a)] $q(rs)=(qr)s$ (associative)
	\item [(b)]$q(r+s)=qr+qs.$
	\item [(c)] For each $q\neq 0,$ there exists $r$ such that $qr=1$	
	\item [(d)] If $qr=qs$ then $r=s$ whenever $q\not=0$
\end{enumerate}
The quaternion product is not commutative.

\section{Frames in quaternion Hilbert space}

\begin{definition}\label{D1}\cite{Ad}
Let $V_{H}^{L}$ is a vector space under left multiplication by quaternionic scalars, where $H$ stands for the quaternion algebra. For $f,g,h\in V_{H}^{L}$ and $q\in H,$ the inner product
\begin{eqnarray*}
\left\langle .|.\right\rangle:V_{H}^{L}\times V_{H}^{L}\longrightarrow H
\end{eqnarray*}
satisfies the following properties:
\begin{enumerate}
	\item [(a)]$\overline{\left\langle f|g\right\rangle}=\left\langle g|f\right\rangle$
	\item [(b)]$\left\|f\right\|^{2}=\left\langle f|f\right\rangle>0~$ unless $~f=0,~$ a real norm
	\item [(c)]$\left\langle f|g+h\right\rangle=\left\langle f|g\right\rangle+\left\langle f|h\right\rangle$
	\item [(d)]$\left\langle \textbf{q}f|g\right\rangle=\textbf{q}\left\langle f|g\right\rangle$
	\item[(e)] $\left\langle f|\textbf{q}g\right\rangle=\left\langle f|g\right\rangle\overline{\textbf{q}}.$
\end{enumerate}
\end{definition}
Assume that the space $V_{H}^{L}$  together with $\left\langle .|.\right\rangle$ is a separable Hilbert space. Properties of left quaternion Hilbert spaces as needed here can be listed as follows: 
For $f,g\in V_{H}^{L}$ and $p,q\in H,$ we have
\begin{enumerate}
	\item [(a)] $pf+qg\in V_{H}^{L}$
	\item[(b)] $p(f+g)=pf+qg$
	\item[(c)] $(pq)f=p(qf)$
	\item [(d)]$(p+q)f=pf+qf.$
\end{enumerate}
\begin{proposition}\label{P1}\cite{Ad} (Schwartz inequality)
$\left\langle f|g\right\rangle\left\langle g|f\right\rangle\leq\left\|f\right\|^{2}\left\|g\right\|^{2},\text{for all}~ f,g\in V_{H}^{L}$
\end{proposition}

For an enhanced explanation of quaternions and quaternion Hilbert spaces one may consult \cite{Ad, ThiAli} and the many references listed there.
\subsection{Some basic facts in left quaternion Hilbert space}
\begin{definition}\label{D2}(\textit{Basis})
Let $V^{L}_{H}$ be a finite dimensional left quaternion Hilbert space with an inner product $\left\langle .| .\right\rangle.$ The linearity is assumed in the second entry of the inner product.
For $V^{L}_{H}$ a sequence $\left\{e_{k}\right\}^{m}_{k=1}$ in $V^{L}_{H}$ is called a basis if the sequence satisfies the following two conditions
\begin{enumerate}
	\item $~V^{L}_{H}=\text{left span}\left\{e_{k}\right\}^{m}_{k=1};~$
	\item $~\left\{e_{k}\right\}^{m}_{k=1}$ is a linearly independent set.
\end{enumerate}

\end{definition}
As a consequence of this definition, for every $f\in V^{L}_{H}$ there exist unique scalar coefficients $\left\{c_{k}\right\}^{m}_{k=1}\subseteq H$ such that $f=\displaystyle\sum_{k=1}^{m}c_{k}e_{k}.$ 
If $\left\{e_{k}\right\}^{m}_{k=1}$ is an orthonormal basis, that is $\left\langle e_{k}|e_{j}\right\rangle=\delta_{kj},$ then 
$\left\langle f|e_{j}\right\rangle=\left\langle \displaystyle\sum_{k=1}^{m}c_{k}e_{k}|e_{j}\right\rangle=\displaystyle\sum_{k=1}^{m}c_{k}\left\langle e_{k}|e_{j}\right\rangle=c_{j}.$
Thereby
\begin{equation}\label{eq4}
	f=\displaystyle\sum_{k=1}^{m}\left\langle f|e_{k}\right\rangle e_{k}.
\end{equation}
We now introduce frames on finite dimensional left quaternion Hilbert spaces. We shall show that the complex treatment adapt to the quaternions as well. In this note {\it left span} means left span over the quaternion scalar field, $H$. We shall also prove the functional analytic properties for quaternions as needed here, and these proofs are the adaptation of the proofs of the complex cases given in \cite{Role}. The theory of frames offered here, more or less, follows the lines of \cite{Ole}.
\begin{definition}\label{D3}(\textit{Frames})
A countable family of elements $\left\{f_{k}\right\}_{k\in I}$ in $V^{L}_{H}$ is a frame for $V^{L}_{H}$ if there exist constants $A,B>0$ such that
\begin{equation}\label{eq5}
A\left\|f\right\|^{2}\leq\displaystyle\sum_{k\in I}\left|\left\langle f|f_{k}\right\rangle\right|^{2}\leq B\left\|f\right\|^{2},
\end{equation}
$~\text{for all}~ f\in V^{L}_{H}.$
\end{definition} 
The real numbers $A$ and $B$ are called frame bounds. The numbers $A$ and $B$ are not unique. The supremum over all lower frame bounds is the \textit{optimal lower frame bound }, and the  infimum over all upper frame bounds is the \textit{optimal upper frame bound} \cite{Ole}. In fact, optimal frame bounds are the frame bounds.
A frame is said to be normalized if $\left\|f_{k}\right\|=1,~\mbox{~~\text{for all}~} k\in I.$ 
In this note we shall only consider finite frames $\left\{f_{k}\right\}^{m}_{k=1},~m\in \mathbb{N.}$
With this restriction, Schwartz inequality shows that
\begin{equation}\label{eq6}
\sum_{k=1}^{m}\left|\left\langle f|f_{k}\right\rangle\right|^{2}\leq\sum_{k=1}^{m}\left\|f_{k}\right\|^{2}\left\|f\right\|^{2},
\quad \text{for all}~ f\in V^{L}_{H}.
\end{equation}
From (\ref{eq6}) it is clear that the upper frame condition is always satisfied with
 $A=\sum_{k=1}^{m}\left\|f_{k}\right\|^{2}$. In order for the lower condition in (\ref{eq5}) to be satisfied, it is necessary that $\text{leftspan}\{f_k\}_{k=1}^{m}=V_H^L$. Let us see this in the following.

\begin{lemma}\label{l1}
Let $\left\{f_{k}\right\}^{m}_{k=1}$ be a sequence in $V^{L}_{H}$ and $W:=left span\left\{f_{k}\right\}^{m}_{k=1}.$ Then a mapping $\varphi:W\longrightarrow \mathbb{R}$ is continuous if and only if for any sequence $\left\{x_{n}\right\}$ in $W$ which converges to $x_{0}$ as $n\rightarrow \infty$ then $\varphi(x_{n})$ converges to  $\varphi(x_{0})$ as $n\rightarrow \infty.$
\end{lemma}
\begin{proof}Similar to the complex case.
\end{proof}
\begin{lemma}\label{l2}
Let $\varphi:W\longrightarrow \mathbb{R}$ be a continuous mapping and $M$ is a compact subset of $W,$ then $\varphi$ assumes a maximum and a minimum at some points of $M.$
\end{lemma}
\begin{proof} Similar to the complex case.
\end{proof}

\begin{proposition}\label{P2}
Let $\left\{f_{k}\right\}^{m}_{k=1}$ be a sequence in $V^{L}_{H}.$ Then $\left\{f_{k}\right\}^{m}_{k=1}$ is a frame for $\text{\textit{left span}}\left\{f_{k}\right\}^{m}_{k=1}.$
\end{proposition}
\begin{proof}
From the Schwartz inequality the upper frame condition is satisfied with $B=\displaystyle\sum_{k=1}^{m}\left\|f_{k}\right\|^{2}.$
Thereby
\begin{equation}\label{eq12} 
\displaystyle\sum_{k=1}^{m}\left|\left\langle f|f_{k}\right\rangle\right|^{2}\leq B\left\|f\right\|^{2}.
\end{equation}
Let $W:=left span\left\{f_{k}\right\}^{m}_{k=1}$ and consider the mapping
\begin{eqnarray*}
\varphi:W\longrightarrow \mathbb{R}, ~\varphi(f):=\displaystyle\sum_{k=1}^{m}\left|\left\langle f,f_{k}\right\rangle\right|^{2}.
\end{eqnarray*}
Now we want to prove that $\varphi$ is continuous.
Let $\left\{g_{n}\right\}$ be a sequence in $W$ such that $g_{n}\longrightarrow g$ as $n\longrightarrow\infty.$
Now,
\begin{eqnarray*}
\left\|\varphi(g_{n})-\varphi(g)\right\|&=&\left\|\sum_{k=1}^{m}\left|\left\langle g_{n}|f_{k}\right\rangle\right|^{2}-\sum_{k=1}^{m}\left|\left\langle g|f_{k}\right\rangle\right|^{2}\right\|\\
&\leq&\sum_{k=1}^{m}\left\|\left|\left\langle g_{n}|f_{k}\right\rangle\right|^{2}-\left|\left\langle g|f_{k}\right\rangle\right|^{2}\right\|\\
&=&\sum_{k=1}^{m}\left\|\left\langle g_{n}|f_{k}\right\rangle\overline{\left\langle g_{n}|f_{k}\right\rangle}-\left\langle g|f_{k}\right\rangle\overline{\left\langle g|f_{k}\right\rangle}\right\|\\
&=&\sum_{k=1}^{m}\left\|\left\langle g_{n}|f_{k}\right\rangle\left\langle f_{k}|g_{n}\right\rangle-\left\langle g|f_{k}\right\rangle\left\langle f_{k}|g\right\rangle\right\|\\
&=&\sum_{k=1}^{m}\left\|\left\langle g_{n}|f_{k}\right\rangle\left\langle f_{k}|g_{n}\right\rangle-\left\langle g|f_{k}\right\rangle\left\langle f_{k}|g_{n}\right\rangle+\left\langle g|f_{k}\right\rangle\left\langle f_{k}|g_{n}\right\rangle-\left\langle g|f_{k}\right\rangle\left\langle f_{k}|g\right\rangle\right\|\\
&=&\sum_{k=1}^{m}\left\|\left(\left\langle g_{n}|f_{k}\right\rangle-\left\langle g|f_{k}\right\rangle\right)\left\langle f_{k}|g_{n}\right\rangle+\left\langle g|f_{k}\right\rangle\left(\left\langle f_{k}|g_{n}\right\rangle- \left\langle f_{k}|g\right\rangle\right)\right\|\\
&=&\sum_{k=1}^{m}\left\|\left(\overline{\left\langle f_{k}| g_{n}\right\rangle}-\overline{\left\langle f_{k}|g\right\rangle}\right)\left\langle f_{k}|g_{n}\right\rangle+\left\langle g|f_{k}\right\rangle\left(\left\langle f_{k}|g_{n}\right\rangle- \left\langle f_{k}|g\right\rangle\right)\right\|\\
&=&\sum_{k=1}^{m}\left\|\left(\overline{\left\langle f_{k}| g_{n}\right\rangle-\left\langle f_{k}|g\right\rangle}\right)\left\langle f_{k}|g_{n}\right\rangle+\left\langle g|f_{k}\right\rangle\left(\left\langle f_{k}|g_{n}\right\rangle- \left\langle f_{k}|g\right\rangle\right)\right\|\\
&=&\sum_{k=1}^{m}\left\|\overline{\left\langle f_{k}| g_{n}-g\right\rangle}\left\langle f_{k}|g_{n}\right\rangle+\left\langle g|f_{k}\right\rangle\left(\left\langle f_{k}|g_{n}-g\right\rangle\right)\right\|\\
&\longrightarrow&0\mbox{~~as~}n\longrightarrow \infty~~\left\lceil \because g_{n}\longrightarrow g \mbox{~as~} n\longrightarrow \infty \right\rceil
\end{eqnarray*}
Thereby $\varphi(g_{n})$ converges to $\varphi(g )$ as $ n\longrightarrow\infty.$
From the lemma (\ref{l1}), $\varphi$ is continuous.
Since the closed unit ball in $W$ is compact, from the lemma (\ref{l2}), we can find $g\in W$ with $\left\|g\right\|=1$ such that
\begin{eqnarray*}
A:=\sum_{k=1}^{m}\left|\left\langle g|f_{k}\right\rangle\right|^{2}=inf\left\{\displaystyle\sum_{k=1}^{m}\left|\left\langle f|f_{k}\right\rangle\right|^{2}:f\in W,\left\|f\right\|=1\right\}.
\end{eqnarray*}
It is clear that $A>0$ as not all $f_{k}$ are zero.
Now given $f\in W,~f\neq 0,~$ we have $\left\|\displaystyle\displaystyle\frac{f}{ \left\|f\right\|}\right\|=1,$ so $\displaystyle\sum_{k=1}^{m}\left|\left\langle \displaystyle\displaystyle\frac{f}{\left\|f\right\|}|f_{k}\right\rangle\right|^{2}\geq A.$
Hence
\begin{eqnarray*}
\displaystyle\sum_{k=1}^{m}\left|\left\langle f|f_{k}\right\rangle\right|^{2}=\displaystyle\sum_{k=1}^{m}\left|\left\langle \displaystyle\displaystyle\frac{f}{\left\|f\right\|}|f_{k}\right\rangle\right|^{2}\left\|f\right\|^{2}\geq A\left\|f\right\|^{2}.
\end{eqnarray*}
Thereby
\begin{equation}\label{eq13}
	\displaystyle\sum_{k=1}^{m}\left|\left\langle f|f_{k}\right\rangle\right|^{2}\geq A\left\|f\right\|^{2}.
\end{equation}
 From (\ref{eq12}) and (\ref{eq13}) ,
\begin{eqnarray*}
A\left\|f\right\|^{2}\leq\displaystyle\sum_{k=1}^{m}\left|\left\langle f|f_{k}\right\rangle\right|^{2}\leq B\left\|f\right\|^{2},~\text{for all ~}f\in W .
\end{eqnarray*}
Hence $\left\{f_{k}\right\}^{m}_{k=1}$ is a frame for $left span\left\{f_{k}\right\}^{m}_{k=1}.$
\end{proof}

\begin{corollary}\label{c1}
A family of elements $\left\{f_{k}\right\}^{m}_{k=1}$ in $V^{L}_{H}$ is a frame for $V^{L}_{H}$ if and only if $left span\left\{f_{k}\right\}^{m}_{k=1}=V^{L}_{H}.$
\end{corollary}
\begin{proof}
Suppose that $\{f_k\}_{k=1}^{m}$ is a frame for $V_H^L$. Then there exist $A,B>0$ such that
\begin{equation}\label{ee1}
A\|f\|^2\leq\sum_{k=1}^{m}\left\vert\langle f\vert f_k\rangle\right\vert^2\leq B\|f\|^2,
\end{equation}
If $ V_H^L\not= leftspan\{f_k\}_{k=0}^{\infty}$ then 
there exists $f\in V_H^L$, $f\not=0$ such that 
$$\langle f|f_k\rangle=0,\quad k=1,\cdots m.$$
Hence $\{f_k\}_{k=1}^m$ cannot be a frame.\\
Conversely suppose that $V_H^L=leftspan\{f_k\}_{k=1}^{m}$. From proposition (\ref{P2}) $\{f_k\}_{k=1}^m$ is a frame for $leftspan\{f_k\}_{k=1}^{m}$, thereby $\{f_k\}_{k=1}^{m}$ is a frame for $V_H^L$.
\end{proof}
From the above corollary it is clear that a frame is an over complete family of vectors in a finite dimensional Hilbert space.
\subsection{Frame operator in left quaternion Hilbert space}
\subsubsection{Frame operators}
Consider now a left quaternion Hilbert space, $V_H^L$ with a frame $\left\{f_{k}\right\}^{m}_{k=1}$ and define a linear mapping 
\begin{equation}\label{eq15}
	T:H^{m}\longrightarrow V^{L}_{H},~T\left\{c_{k}\right\}^{m}_{k=1}=\displaystyle\sum_{k=1}^{m} c_{k}f_{k},~c_{k}\in H.
\end{equation}
$T$ is called the  \textit{synthesis operator} or \textit{pre-frame operator}.
 The adjoint of $T$
\begin{equation}\label{eq16}
T^{*}:V^{L}_{H}\longrightarrow H^{m}, \mbox{~~by~~}T^{*}f=\left\{\left\langle f|f_{k}\right\rangle\right\}^{m}_{k=1}	
\end{equation}
is called the \textit{analysis operator.}
 By composing $T$ with its adjoint we obtain the \textit{frame operator}
\begin{equation}\label{eq17}
S:V^{L}_{H}\longrightarrow V^{L}_{H},~Sf=TT^{*}f=\displaystyle\sum_{k=1}^{m}\left\langle f|f_{k}\right\rangle f_{k}	.
\end{equation}
Note that in terms of the frame operator, for $f\in V^{L}_{H}$
\begin{eqnarray*}
\left\langle Sf|f\right\rangle &=&\left\langle \sum_{k=1}^{m}\left\langle f|f_{k}\right\rangle f_{k}|f\right\rangle
=\sum_{k=1}^{m}\left\langle f|f_{k}\right\rangle\left\langle f_{k}|f\right\rangle
=\sum_{k=1}^{m}\left|\left\langle f|f_{k}\right\rangle\right|^{2}.
\end{eqnarray*}
That is,
\begin{equation}\label{eq18}
	\left\langle Sf|f\right\rangle=\displaystyle\sum_{k=1}^{m}\left|\left\langle f|f_{k}\right\rangle\right|^{2},~f\in V^{L}_{H}.
\end{equation}
A frame $\left\{f_{k}\right\}^{m}_{k=1}$ is tight if we can choose $A=B$ in the definition (\ref{D3}), in this case (\ref{eq5})  gives
\begin{eqnarray*}
\displaystyle\sum_{k=1}^{m}\left|\left\langle f|f_{k}\right\rangle\right|^{2}=A\left\|f\right\|^{2},~\text{for all ~}f\in V^{L}_{H}. 
\end{eqnarray*}
Thereby $\left\langle Sf|f\right\rangle=A\left\|f\right\|^{2},~\text{for all ~}f\in V^{L}_{H}. $

\begin{proposition}\label{P3}
Let $\left\{f_{k}\right\}^{m}_{k=1}$ be a tight frame for $V^{L}_{H}$ with the frame bound $A.$ Then $S=AI$ (where $I$ is the identity operator on $V^{L}_{H}$),and 
\begin{eqnarray*}
f=\frac{1}{A}\sum_{k=1}^{m}\left\langle f|f_{k}\right\rangle f_{k},~\text{for all ~}f\in V^{L}_{H}. 
\end{eqnarray*}
\end{proposition}
\begin{proof}
The frame operator $S$ is given by
\begin{eqnarray*}
S: V^{L}_{H}\longrightarrow  V^{L}_{H};\quad Sf=\sum_{k=1}^{m}\left\langle f|f_{k}\right\rangle f_{k}.
\end{eqnarray*}
Let $f\in V^{L}_{H},$ then
\begin{eqnarray*}
\left\langle Sf|f\right\rangle &=&\left\langle \sum_{k=1}^{m}\left\langle f|f_{k}\right\rangle f_{k}|f\right\rangle
=\displaystyle\sum_{k=1}^{m}\left|\left\langle f|f_{k}\right\rangle\right|^{2}.
\end{eqnarray*}
Since the frame $\left\{f_{k}\right\}^{m}_{k=1}$ is tight for $V_{H}^{L},$
\begin{eqnarray*}
\left\langle Sf|f\right\rangle =A\left\|f\right\|^{2},~\text{for all ~}f\in V^{L}_{H}.
\end{eqnarray*}
Now,
\begin{eqnarray*}
\left\langle Sf|f\right\rangle &=&A\left\|f\right\|^{2}
=A\left\langle f|f\right\rangle
=A\left\langle If|f\right\rangle
=\left\langle AIf|f\right\rangle.
\end{eqnarray*}
Thereby $\left\langle Sf|f\right\rangle=\left\langle AIf|f\right\rangle$, for all $f\in V^{L}_{H}.$
Hence $S=AI.$
Since $\left\{f_{k}\right\}^{m}_{k=1}$ is a frame for $V^{L}_{H},$ from corollary (\ref{c1} ),
\begin{equation}\label{eq19}
V^{L}_{H}=leftspan\left\{f_{k}\right\}^{m}_{k=1}.
\end{equation}
Therefore for given $f\in V^{L}_{H},$ there exists $c_{k}\in H$ such that 
\begin{equation}\label{eq20}
	f=\sum_{k=1}^{m}c_{k}f_{k}.
\end{equation}
Now define $c_{k}=\left\langle f|g_{k}\right\rangle$ and $g_{k}=\displaystyle\frac{1}{A}f_{k},$ here $g_{k}\in V^{L}_{H}.$
Then (\ref{eq20} ) becomes
\begin{eqnarray*}
f &=&\sum_{k=1}^{m}\left\langle f|g_{k}\right\rangle f_{k} 
=\sum_{k=1}^{m}\left\langle f|\frac{1}{A}f_{k}\right\rangle f_{k}
=\sum_{k=1}^{m}\left\langle f|f_{k}\right\rangle \overline{\left(\frac{1}{A}\right)} f_{k}
=\frac{1}{A}\sum_{k=1}^{m}\left\langle f|f_{k}\right\rangle f_{k}\mbox{~~~as~} A \mbox{~~is real}.
\end{eqnarray*}
Hence 
$$f=\frac{1}{A}\sum_{k=1}^{m}\left\langle f|f_{k}\right\rangle f_{k},\mbox{~~\text{for all}~} f\in V^{L}_{H}. $$
\end{proof}

\begin{definition}\cite{Ad}
Let $V^L_H$ be any left quaternion Hilbert space. A mapping $S:V^L_H\longrightarrow V^L_H$ is said to be \textit{left-linear} if,
\begin{eqnarray*}
S(\alpha f+\beta g)=\alpha S(f)+\beta S(g),
\end{eqnarray*}
for all $f,g\in V^L_H$ and $\alpha,\beta\in H.$
\end{definition}
\begin{definition}\cite{Ad}
A linear operator $S:V^L_H\longrightarrow V^L_H$ is said to be \textit{bounded} if,
\begin{eqnarray*}
\left\|Sf\right\|\leq K\left\|f\right\|,
\end{eqnarray*}
for some constant $K\geq 0$ and all $f\in V^L_ H.$
\end{definition}
\begin{definition}\cite{Ad}\label{D4}(Adjoint operator )
Let $S:V_{H}^{L}\longrightarrow V_{H}^{L}$ be a bounded linear operator on a left quaternion Hilbert space $V_{H}^{L}.$ We define its adjoint to be the operator
 $S^*:V_{H}^{L}\longrightarrow V_{H}^{L}$ that has the property
\begin{equation}\label{eq21}
\left\langle f|Sg\right\rangle=\left\langle S^{*}f|g\right\rangle,~\text{for all~}f,g\in V_{H}^{L}.
\end{equation}
\end{definition}

\begin{lemma}\label{l3}
The adjoint operator $S^*$ of a bounded linear operator is linear and bounded.
\end{lemma}
\begin{proof}Linearity can easily be verified and every linear operator on a finite dimensional quaternion Hilbert space is bounded.
\end{proof}

\begin{definition}\cite{Ad}\label{D5}(\textit{Self-adjoint operator})Let $V_{H}^{L}$ be a left quaternion Hilbert space. A bounded linear operator $S$ on
$V_{H}^{L}$ is called \textit{self-adjoint, }if $S=S^*.$
\end{definition}

\begin{lemma}\label{l4}
Let $S:V_{H}^{L}\longrightarrow V_{H}^{L}~$and$~T:V_{H}^{L}\longrightarrow V_{H}^{L}~$be bounded linear operators on $V_{H}^{L}.$ Then for any$~f,g\in V_{H}^{L},$  we have
\begin{enumerate}
	\item [(a)] $\left\langle Sf|g\right\rangle=\left\langle f|S^{*}g\right\rangle.$
	\item [(b)]$(S+T)^{*}=S^{*}+T^{*}.$
	\item [(c)] $(ST)^{*}=T^{*}S^{*}.$ 	
	\item [(d)] $(S^{*})^{*}=S.$
	\item [(e)] $I^{*}=I,~$ where $~I~$ is an identity operator on $~V_{H}^{L}.~$
	\item [(f)] If $~S~$ is invertible then $~(S^{-1})^*=(S^{*})^{-1}.~$
\end{enumerate}
\end{lemma}
\begin{proof}It is straightforward.
\end{proof}

\begin{lemma}\label{l5}
Let $U_{H}^{L}$ and $V_{H}^{L}$ be finite dimensional left quaternion Hilbert spaces and $S:U_{H}^{L}\longrightarrow V_{H}^{L}$ be a linear mapping then $\ker S$ is a subspace of $U_{H}^{L}.$
\end{lemma}
\begin{proof}
Easy to verify.
\end{proof}

\begin{lemma}\label{l6}
Let $U_{H}^{L},V_{H}^{L}$ be finite dimensional left quaternion Hilbert spaces and $S:U_{H}^{L}\longrightarrow V_{H}^{L}$ be a linear mapping then 
\begin{eqnarray*}
\text{dim}R_{S}+\text{dim}N_{S}=\text{dim}U_{H}^{L},
\end{eqnarray*}
where $R_{S}:=\text{image of}~ S,~N_{S}:=\ker S.$
\end{lemma}
\begin{proof}Proof from the complex theory can easily be adapted.
\end{proof}

\begin{lemma}\label{l7}
Let $S:U_{H}^{L}\longrightarrow V_{H}^{L}$ be a linear mapping. $S$ is one to one if and only if  $N_{S}=\left\{0\right\}.$
\end{lemma}
\begin{proof}
Easy to verify.
\end{proof}

\begin{lemma}\label{l8}
Let $U_{H}^{L},V_{H}^{L}$ are finite dimensional left quaternion Hilbert spaces with same dimension.
Let $S:U_{H}^{L}\longrightarrow V_{H}^{L}$ be a linear mapping. If $S$ is one to one then $S$ is onto.
\end{lemma}
\begin{proof}
Easy to verify.
\end{proof}

\begin{lemma}\label{l9}(Pythagoras' law)
Suppose that $f$ and $g$ is an arbitrary pair of orthogonal vectors in the left quaternion  Hilbert space $V_{H}^{L}.$ Then we have Pythagoras' formula
\begin{equation}\label{eq25}
	\left\|f+g\right\|^{2}=\left\|f\right\|^{2}+\left\|g\right\|^{2}.
\end{equation}
\end{lemma}
\begin{proof}
Straightforward.  
\end{proof}

\begin{lemma}\label{l10}
Let $T:H^{m}\longrightarrow V^{L}_{H}$ be a linear mapping and $T^{*}:V^{L}_{H}\longrightarrow H^{m}$ be its adjoint operator. Then $N_{T}=R^{\bot}_{T^{*}},$ where $N_{T}:=\ker T$ and $R_{T^{*}}:=\mbox{~range of~}T^{*}.$ 
\end{lemma}
\begin{proof}
Proof from the complex theory can easily be adapted.
\end{proof}

\begin{theorem}\label{T1}
Let $\left\{f_{k}\right\}^{m}_{k=1}$ be a frame for $V^{L}_{H}$ with frame operator $S.$ Then
\begin{enumerate}
\item $S$ is invertible and self-adjoint.
\item Every $f\in V_{H}^{L},$ can be represented as
$$f=\sum_{k=1}^{m}\left\langle f|S^{-1}f_{k}\right\rangle f_{k}=\sum_{k=1}^{m}\left\langle f|f_{k}\right\rangle S^{-1} f_{k}.$$
\item If $f\in V_{H}^{L},$ and has the representation $f=\displaystyle\sum_{k=1}^{m}c_{k}f_{k}$ for some scalar coefficients $\left\{c_{k}\right\}^{m}_{k=1},$ then $$\sum_{k=1}^{m}\left|c_{k}\right|^{2}=\sum_{k=1}^{m}\left|\left\langle f|S^{-1}f_{k}\right\rangle\right|^{2}+\sum_{k=1}^{m}\left|c_{k}-\left\langle f|S^{-1}f_{k}\right\rangle\right|^{2}.$$
\end{enumerate}
\end{theorem}
\begin{proof}

(1)~ $S:V_{H}^{L}\longrightarrow V_{H}^{L},$ by $Sf=TT^{*}f=\displaystyle\sum_{k=1}^{m}\left\langle f|f_{k}\right\rangle f_{k},$ $\text{for all ~}f\in V_{H}^{L}.$
	Now	
	\begin{eqnarray*}
S^*=(TT^*)^*
=(T^*)^*T^*
=TT^*
=S.
	\end{eqnarray*}
It follows that $S$ is self-adjoint.
We have $\ker S=\left\{f~:~Sf=0\right\}.$
Let $f\in \ker S,$ then $Sf=0.$
Therefore 
	\begin{eqnarray*}
	0&=&\left\langle Sf|f\right\rangle
	=\left\langle\displaystyle\sum_{k=1}^{m}\left\langle f|f_{k}\right\rangle f_{k}|f \right\rangle
	=\sum_{k=1}^{m}\left|\left\langle f|f_{k}\right\rangle \right|^{2}.
	\end{eqnarray*}	
	Thereby $\sum_{k=1}^{m}\left|\left\langle f|f_{k}\right\rangle \right|^{2}=0.$
	Since $\left\{f_{k}\right\}^{m}_{k=1}$ be a frame for $V^{L}_{H},$ by definition-\ref{D3},
	$$A\left\|f\right\|^{2}\leq\sum_{k=1}^{m}\left|\left\langle f|f_{k}\right\rangle \right|^{2}\leq B\left\|f\right\|^{2},~\text{for all ~}f\in V^{L}_{H}.$$
Hence $A\left\|f\right\|^{2}\leq 0\leq B\left\|f\right\|^{2},$ $\text{for all ~}f\in V^{L}_{H}$ and $A,B>0.$
So $\left\|f\right\|^{2}=0.$ 
Therefore $f=0,$ which implies  $N_{S}={0}.$
Hence $S$ is one to one.
Since $V^L_H$ is of  finite dimension, from the lemma  (\ref{l8} ) $S$ is onto .
Therefore  $S$ is invertible.\\
(2)~ If $S$ is self-adjoint then $S^{-1}$ is self-adjoint. For,
Consider 
\begin{eqnarray*}
(S^{-1})^{*}=(S^*)^{-1}
=(S)^{-1}~~\text{as $~S~$ is self adjoint}.
\end{eqnarray*}
Thereby $S^{-1}$ is self-adjoint.
If $S:V_{H}^{L}\longrightarrow V_{H}^{L},$ is linear and bijection then $S^{-1}$ is linear. For,
Since $S$ is onto, $S^{-1}:V_{H}^{L}\longrightarrow V_{H}^{L}$.
Let $f,g\in V^{L}_{H}$ then there exists $k,h\in V^{L}_{H}$ such that 
$S^{-1}(f)=k~\text{and }~S^{-1}(g)=h.$
Thereby $f=S(k)~\text{and}~g=S(h).$
Let $\alpha,\beta\in H,$then
\begin{eqnarray*}
S^{-1}(\alpha f+\beta g)&=&S^{-1}(\alpha S(k)+\beta S(h))\\
&=&S^{-1}( S(\alpha k+\beta h))~~\text {as $~S~$ is linear.}\\
&=&\alpha k+\beta h\\
&=&\alpha S^{-1}(f)+\beta S^{-1}(g)
\end{eqnarray*}
Thereby for all $f,g\in V^{L}_{H}$ and $\alpha,\beta\in H,$
\begin{eqnarray*}
S^{-1}(\alpha f+\beta g)=\alpha S^{-1}(f)+\beta S^{-1}(g)
\end{eqnarray*}
Hence $S^{-1}$ is linear.
 Let $f\in V^{L}_{H},$
then
\begin{eqnarray*}
f&=&SS^{-1}f
=TT^*S^{-1}f
=\sum_{k=1}^m\left\langle S^{-1}f|f_k\right\rangle f_k
=\sum_{k=1}^m\left\langle f|(S^{-1})^*f_k\right\rangle f_k,
=\sum_{k=1}^m\left\langle f|S^{-1}f_k\right\rangle f_k. 
\end{eqnarray*}
Thereby for every $f\in V^{L}_{H},$
\begin{equation}\label{eq33}
 ~	f=\sum_{k=1}^m \left\langle f|S^{-1}f_k\right\rangle  f_k.
\end{equation}
Similarly we have
\begin{eqnarray*}
f&=&S^{-1}Sf
=S^{-1}TT^*f
=S^{-1}\left(\sum_{k=1}^m\left\langle f|f_k\right\rangle f_k\right)
=\sum_{k=1}^m S^{-1}\left(\left\langle f|f_k\right\rangle f_k\right)
=\sum_{k=1}^m \left\langle f|f_k\right\rangle S^{-1} f_k.
\end{eqnarray*}
Thereby for every $ f\in V^{L}_{H},$
\begin{equation}\label{eq34}
	f=\sum_{k=1}^m \left\langle f|f_k\right\rangle S^{-1} f_k.
\end{equation}
From  (\ref{eq33} ) and  (\ref{eq34} ),$\mbox{~for every~} f\in V^{L}_{H},$
\begin{eqnarray*}
	f=\sum_{k=1}^m \left\langle f|S^{-1}f_k\right\rangle  f_k=\sum_{k=1}^m \left\langle f|f_k\right\rangle S^{-1} f_k.
\end{eqnarray*}
(3)~ Let $f\in V^{L}_{H},$ from corollary  (\ref{c1} )  
$$f=\sum_{k=1}^{m}c_{k}f_{k},~ \mbox{~~for some}~c_{k}\in H.$$\\
From the part(1), 
\begin{equation}\label{eq35}
	f=\sum_{k=1}^{m}c_{k}f_{k}=\sum_{k=1}^{m}\left\langle f|S^{-1}f_{k}\right\rangle f_{k}.
\end{equation}
Hence
\begin{equation}\label{eq36}
\sum_{k=1}^{m}\left(c_{k}-\left\langle f|S^{-1}f_{k}\right\rangle \right)f_{k}=0.	
\end{equation}
Thereby $\displaystyle\sum_{k=1}^{m}d_{k}f_{k}=0,$ for some $d_{k}=\left(c_{k}-\left\langle f|S^{-1}f_{k}\right\rangle \right)\in H.$
From  (\ref{eq15} ), $ T:H^{m}\longrightarrow V^{L}_{H},$ is defined by $T\left\{d_{k}\right\}^{m}_{k=1}=\displaystyle\sum_{k=1}^{m} d_{k}f_{k},~d_{k}\in H.$ 
We have $N_{T}=\left\{\left\{d_{k}\right\}_{k=1}^m~|~T\left\{d_{k}\right\}^{m}_{k=1}=0\right\},$ therefore
\begin{equation}\label{eq37}
	\left\{d_{k}\right\}_{k=1}^{m}=\left\{c_{k}\right\}^{m}_{k=1}-\left\{\left\langle f|S^{-1}f_{k}\right\rangle\right\}^{m}_{k=1} \in N_{T}.
\end{equation}
From lemma  (\ref{l10} ),  $N_{T}=R^{\bot}_{T^{*}},$ then 
\begin{equation}\label{eq38}
	\left\{c_{k}\right\}^{m}_{k=1}-\left\{\left\langle f|S^{-1}f_{k}\right\rangle\right\}^{m}_{k=1} \in R^{\bot}_{T^{*}}.
\end{equation}
From  (\ref{eq15} ) and  (\ref{eq16} ) we have $T^{*}:V^{L}_{H}\longrightarrow H^{m},~T^{*}f=\left\{\left\langle f|f_{k}\right\rangle\right\}^{m}_{k=1}, $
and $S:V^{L}_{H}\longrightarrow V^{L}_{H},~Sf=TT^{*}f=\displaystyle\sum_{k=1}^{m}\left\langle f|f_{k}\right\rangle f_{k}.$
Hence $T^{*}(S^{-1}f)=\left\{\left\langle S^{-1}f|f_{k}\right\rangle\right\}^{m}_{k=1}.$
Therefore
$$\left\{\left\langle S^{-1}f|f_{k}\right\rangle\right\}^{m}_{k=1}\in R_{T^{*}}.$$
Since $S^{-1}$ is self adjoint,
$$\left\{\left\langle S^{-1}f|f_{k}\right\rangle\right\}^{m}_{k=1}=\left\{\left\langle f|S^{-1}f_{k}\right\rangle\right\}^{m}_{k=1}.$$
Hence
\begin{equation}\label{eq39}
\left\{\left\langle f|S^{-1}f_{k}\right\rangle\right\}^{m}_{k=1}=\left\{\left\langle S^{-1}f|f_{k}\right\rangle\right\}^{m}_{k=1}\in R_{T^{*}}.	
\end{equation}
Now we can write, 
\begin{equation}\label{eq40}
	\left\{c_{k}\right\}^{m}_{k=1}=\left\{c_{k}\right\}^{m}_{k=1}-\left\{\left\langle f|S^{-1}f_{k}\right\rangle\right\}^{m}_{k=1}+\left\{\left\langle f|S^{-1}f_{k}\right\rangle\right\}^{m}_{k=1}.
\end{equation}
From  (\ref{eq38} ),   (\ref{eq39} ),  (\ref{eq40} ) and lemma  (\ref{l9} ),
\begin{equation}\label{eq41}
\sum_{k=1}^{m}\left|c_{k}\right|^{2}=\sum_{k=1}^{m}\left|\left\langle f|S^{-1}f_{k}\right\rangle\right|^{2}+\sum_{k=1}^{m}\left|c_{k}-\left\langle f|S^{-1}f_{k}\right\rangle\right|^{2}.
\end{equation}
\end{proof}

Theorem (\ref{T1} ) is one of the most important results about frames, and  
\begin{eqnarray*}
	f=\sum_{k=1}^m \left\langle f|S^{-1}f_k\right\rangle  f_k=\sum_{k=1}^m \left\langle f|f_k\right\rangle S^{-1} f_k.
\end{eqnarray*}
is called the \textit{frame decomposition.} If $\left\{f_{k}\right\}^{m}_{k=1}$ is a frame but not a basis, there exists non-zero sequences $\left\{h_{k}\right\}^{m}_{k=1}$ such that $\displaystyle\sum_{k=1}^{m}h_{k}f_{k}=0.$ Thereby $f\in~V^{L}_{H}$ can be written as
\begin{eqnarray*}
	f&=&\sum_{k=1}^m \left\langle f|S^{-1}f_k\right\rangle  f_k+\sum_{k=1}^m h_k f_k
	=\sum_{k=1}^m \left(\left\langle f|S^{-1}f_k\right\rangle +h_k \right)f_k
\end{eqnarray*}
showing that $f$ has many representations as superpositions of the frame elements.

\begin{corollary}\label{c2}
Assume that $\left\{f_{k}\right\}^{m}_{k=1}$ is a basis for $V^{L}_{H}.$ Then there exists a unique family $\left\{g_{k}\right\}^{m}_{k=1}$ in $V^{L}_{H}$ such that 
\begin{equation}\label{eq42}
f=\sum_{k=1}^{m}\left\langle f|g_{k}\right\rangle f_{k},~\text{for all ~}f\in V^{L}_{H}.	
\end{equation}
In terms of the frame operator, $\left\{g_{k}\right\}^{m}_{k=1}=\left\{S^{-1}f_{k}\right\}^{m}_{k=1}.$ Furthermore $\left\langle f_{j}|g_{k}\right\rangle=\delta_{j,k}.$
\end{corollary}
\begin{proof}
Let $f\in V^{L}_{H},$ from the Theorem  (\ref{T1} ),
\begin{equation}\label{eq43}
f=\sum_{k=1}^{m}\left\langle f|S^{-1}f_{k}\right\rangle f_{k}.	
\end{equation}
Now take $\left\{g_{k}\right\}^{m}_{k=1}=\left\{S^{-1}f_{k}\right\}^{m}_{k=1},$ in  (\ref{eq43} ) then,
\begin{equation}\label{eq44}
	f=\sum_{k=1}^{m}\left\langle f|g_{k}\right\rangle f_{k}.
\end{equation}
Hence there exists a family $\left\{g_{k}\right\}^{m}_{k=1}$ in $V^{L}_{H}$ such that 
\begin{equation}\label{eq45}
f=\sum_{k=1}^{m}\left\langle f|g_{k}\right\rangle f_{k},~\text{for all ~}f\in V^{L}_{H}.
\end{equation}
\textit{Uniqueness:} Assume that there is another family $\left\{h_{k}\right\}^{m}_{k=1} ~$ in $~V^{L}_{H}$ such that $$f=\displaystyle\sum_{k=1}^{m}\left\langle f|h_{k}\right\rangle f_{k},~\text{for all ~}f\in V^{L}_{H}.$$
Then $\displaystyle\sum_{k=1}^{m}\left\langle f|g_{k}\right\rangle f_{k}=\displaystyle\sum_{k=1}^{m}\left\langle f|h_{k}\right\rangle f_{k}.~$\\
$\Longrightarrow \displaystyle\sum_{k=1}^{m}\left(\left\langle f|g_{k}\right\rangle- \left\langle f|h_{k}\right\rangle\right)f_{k}=0.~\\
\Longrightarrow\displaystyle\sum_{k=1}^{m}\left\langle f|g_{k}-h_{k}\right\rangle f_{k}=0.~\\
\Longrightarrow \left\langle f|g_{k}-h_{k}\right\rangle =0,\mbox{~~ for all ~} k=1,2,\cdots,m,\mbox{~~ as ~}\left\{f_{k}\right\}^{m}_{k=1}\mbox{~~ is a basis for ~}V^{L}_{H}.\\
\Longrightarrow \left\langle f|g_{k}-h_{k}\right\rangle =0,\mbox{~~ for all ~} k=1,2,\cdots,m,\mbox{~~ for all ~}f\in V^{L}_{H}.\\
\Longrightarrow  g_{k}-h_{k}=0,\mbox{~~ for all ~} k=1,2,\cdots,m.\\
\Longrightarrow  g_{k}=h_{k},\mbox{~~ for all ~} k=1,2,\cdots,m.$\\
Hence there exists a unique family $\left\{g_{k}\right\}^{m}_{k=1}$ in $V^{L}_{H}$ such that
\begin{equation}\label{eq46}
f=\displaystyle\sum_{k=1}^{m}\left\langle f|g_{k}\right\rangle f_{k},~\text{for all ~}f\in V^{L}_{H}.	
\end{equation}
Since $f=\displaystyle\sum_{k=1}^{m}\left\langle f|g_{k}\right\rangle f_{k},~\text{for all ~}f\in V^{L}_{H},$
for fixed $f_{j}\in V^{L}_{H},$
\begin{equation}\label{eq47}
	f_{j}=\displaystyle\sum_{k=1}^{m}\left\langle f_{j}|g_{k}\right\rangle f_{k}.
\end{equation}
Since $\left\{f_{k}\right\}^{m}_{k=1}$ is a basis for $V^{L}_{H},~\left\{f_{k}\right\}^{m}_{k=1}$ is linearly independent.
Therefore in  (\ref{eq47} ), $\left\langle f_{j}|g_{k}\right\rangle=\delta_{j,k}.$
Otherwise ($\left\langle f_{j}|g_{k}\right\rangle\neq\delta_{j,k}.~$),$~\left\{f_{k}\right\}^{m}_{k=1}$ becomes linearly dependent.
Hence, 
\begin{equation}\label{eq48}
\left\langle f_{j}|g_{k}\right\rangle=\delta_{j,k}.	
\end{equation}
\end{proof}

We can give a perceptive clarification of why frames are important in signal transmission. Let us say we want to transmit a signal $f$ that belonging to a left quaternion Hilbert space from a transmitter $T$ to a receiver $R$. Suppose that both $T $ and $R$ have the knowledge of
frame $\left\{f_{k}\right\}^{m}_{k=1}$ for $V^{L}_{H}$. Let
 $T $ transmits the frame coefficients $\left\{\left\langle f|S^{-1}f_{k}\right\rangle\right\}^{m}_{k=1}.$ Using the received numbers, the receiver $R$ can reconstruct the signal $f$ using the frame decomposition.
If $R$ receives a perturbed (noisy) signal,  $\left\{\left\langle f|S^{-1}f_{k}\right\rangle+c_{k}\right\}^{m}_{k=1}$ of
the correct frame coefficients, using the received coefficients, $R$ will reconstruct the transmitted signal as
\begin{eqnarray*}
	\sum^{m}_{k=1}\left(\left\langle f|S^{-1}f_{k}\right\rangle+c_{k}\right)f_{k}&=&\sum^{m}_{k=1}\left\langle f|S^{-1}f_{k}\right\rangle f_{k}+\sum^{m}_{k=1}c_{k}f_{k}
	=f+\sum^{m}_{k=1}c_{k}f_{k}
\end{eqnarray*}
this differs from the correct signal $f$ by the noise $\sum^{m}_{k=1}c_{k}f_{k}.$  Minimizing this noise for various signals with different types of noises has been a hot topic in signal processing. Since the frame $\{f_k\}_{k=1}^m$ is an over complete set,  it is possible that some part the noise contribution may sum to zero. At the same time, if $\{f_k\}_{k=1}^m$ is an orthonormal basis this scenario is never a possibility. In that case
$$\|\sum_{k=1}^mc_kf_k\|^2=\sum_{k=1}^m|c_k|^2,$$
so each noise contribution will make the reconstruction worse.

\begin{definition}\label{DD1}
For $0<p<\infty$,
$$\ell^p=\left\{x=\{x_n\}\subset H~\vert~\sum_{n}|x_n|^p<\infty\right\}.$$
If $p\geq 1$, $\displaystyle \|x\|_p=\left(\sum_{n}|x_n|^p\right)^{\frac{1}{p}}$
defines a norm in $\ell^p$. In fact $\ell^p$ is a complete metric space with respect to this norm.
\end{definition}
We have already seen that, for $f\in V_H^L$, the frame coefficients $\{\langle f|S^{-1}f_k\rangle\}_{k=1}^m$ have minimal $\ell^2$ norm among all sequences $\{c_k\}_{k=1}^m$ for which $f=\sum_{k=1}^mc_kf_k$. In the next theorem, let us see that the existence of coefficients minimizing the $\ell^1$ norm.
\begin{theorem}\label{T2}
Let $\left\{f_{k}\right\}^{m}_{k=1}$ be a frame for a finite-dimensional left quaternion Hilbert space  $V^{L}_{H}.$ Given $f\in V^{L}_{H},$ there exist coefficients $\left\{d_{k}\right\}^{m}_{k=1}\in H^{m}$ such that
$f=\displaystyle\sum_{k=1}^{m}d_{k}f_{k},$ and 
\begin{equation}\label{eq49}
\sum_{k=1}^{m}\left|d_{k}\right|=\inf\left\{\sum_{k=1}^{m}\left|c_{k}\right|:~f=\sum_{k=1}^{m}c_{k}f_{k}\right\}.
\end{equation}
\end{theorem}
\begin{proof}
Fix $f\in V^{L}_{H}$. It is clear that we can choose a set of coefficients $\left\{c_{k}\right\}^{m}_{k=1},~c_{k}\in  H $ such that  
$f=\displaystyle\sum_{k=1}^{m}c_{k}f_{k}.$
Let $r:=\displaystyle\sum_{k=1}^{m}\left|c_{k}\right|.$
Since we want to minimize the $\ell^{1}$ norm of the coefficients, we can now restrict our search for a minimizer to sequences $\left\{d_{k}\right\}^{m}_{k=1}$ belonging to the compact set
\begin{equation}\label{eq50}
M:=\left\{~\left\{d_{k}\right\}^{m}_{k=1}\in H^{m}~ :~\left|d_{k}\right|\leq r,~k=1,2,...,m.\right\}
\end{equation}
Now,
\begin{equation}\label{eq51}
\left\{~\left\{d_{k}\right\}^{m}_{k=1}\in M~|f=\sum_{k=1}^{m}d_{k}f_{k} \right\}~\mbox{~~ is compact.}	
\end{equation}
Define a function
\begin{equation}\label{eq52}
	\varphi :H^{m}\longrightarrow \mathbb{R},~ \varphi\left\{d_{k}\right\}^{m}_{k=1}:=\displaystyle\sum_{k=1}^{m}\left|d_{k}\right|.
\end{equation}
We can prove $\varphi$ is continuous by similar proof of proposition  (\ref{P2} ).
From  (\ref{eq50} ) and lemma  (\ref{l2} ),
\begin{equation}\label{eq53}
\sum_{k=1}^{m}\left|d_{k}\right|=\inf\left\{\sum_{k=1}^{m}\left|c_{k}\right|:~f=\sum_{k=1}^{m}c_{k}f_{k}\right\}.	
\end{equation}
Hence for given $f\in V^{L}_{H},$ there exist coefficients $\left\{d_{k}\right\}^{m}_{k=1}\in H^{m}$ such that $f=\displaystyle\sum_{k=1}^{m}d_{k}f_{k},$ and 
\begin{equation}\label{eq54}
\sum_{k=1}^{m}\left|d_{k}\right|=\inf\left\{\sum_{k=1}^{m}\left|c_{k}\right|:~f=\sum_{k=1}^{m}c_{k}f_{k}\right\}.
\end{equation}
\end{proof}

Let $W=\text{\textit{left span}}\left\{f_{k}\right\}^{m}_{k=1}$, then in view of Proposition (\ref{P2} ), the set of vectors $\left\{f_{k}\right\}^{m}_{k=1}$ is a frame for $W$. If $W\neq V^{L}_{H},$ then using the frame decomposition of the frame $\left\{f_{k}\right\}^{m}_{k=1}$ one can obtain useful expression for the orthogonal projection onto the subspace $W$.

 \begin{theorem}\label{T3}
Let $\left\{f_{k}\right\}^{m}_{k=1}$ be a frame for a subspace $W$ of the left quaternion Hilbert space  $V^{L}_{H}$. Then the orthogonal projection of   $V^{L}_{H}$ onto $W$ is given by 
\begin{equation}\label{eq55}
Pf=\sum_{k=1}^{m}\left\langle f|S^{-1}f_{k}\right\rangle f_{k}.	
\end{equation}
\end{theorem}
\begin{proof}
Define an operator $P$ from $V^{L}_{H}$ to $W$ by 
\begin{equation}\label{eq56}
P:V^{L}_{H}\longrightarrow W\mbox{~~by~~} ~Pf=\sum_{k=1}^{m}\left\langle f|S^{-1}f_{k}\right\rangle f_{k},\mbox{~for all ~}f\in V^{L}_{H}.	
\end{equation}
First let us prove that $P$ is onto. For,
let $f_{1}\in W $ and $S:W\longrightarrow W$ be a frame operator in $W.$
Since $\left\{f_{k}\right\}^{m}_{k=1}$ be a frame for the subspace $W,$ we have
\begin{eqnarray*}
f_{1}=\sum_{k=1}^{m}\left\langle f_{1}|S^{-1}f_{k}\right\rangle f_{k}
\end{eqnarray*}
But $W\subseteq V^{L}_{H},$ thereby $f_{1}\in V^{L}_{H}.$
Since $f_{1}$ is arbitrary,
for given $f\in W$, there exists $g\in V^{L}_{H}$ such that $Pg=f.$
Thereby $P$ is onto.
Now we want to prove that $P$ is an orthogonal projection.
Hence our claims are 
\begin{enumerate}
	\item[(i)] $Pf=f,~$for $~f\in W.~$
	\item[(ii)] $Pf=0,~$for $~f\in W^{\bot.}~$
\end{enumerate}
For,
\begin{enumerate}
	\item[(i)] The mapping $P:V^{L}_{H}\stackrel{onto}{\longrightarrow}W$ is given by 
\begin{equation}\label{eq57}
	Pf=\sum_{k=1}^{m}\left\langle f|S^{-1}f_{k}\right\rangle f_{k}.
\end{equation}
Since $\left\{f_{k}\right\}^{m}_{k=1}$ be a frame for a subspace $W$ of the left quaternion Hilbert space  $V^{L}_{H},$ from  (\ref{eq15} ) the frame operator $S$ is given by
\begin{equation}\label{eq58}
S:W\longrightarrow W,\mbox{~~by~}Sf=\displaystyle\sum_{k=1}^{m}\left\langle f|f_{k}\right\rangle f_{k},\mbox{~\text{for all}~}	f\in W.
\end{equation}
From the Theorem  (\ref{T1} ), every $f\in W,$ can be represented as
\begin{equation}\label{eq59}
	f=\sum_{k=1}^{m}\left\langle f|S^{-1}f_{k}\right\rangle f_{k}.
\end{equation}
From  (\ref{eq57} )and (\ref{eq59} ),
\begin{equation}\label{eq60}
	Pf=f,~\text{for all ~}f\in W.
\end{equation}

	\item[(ii)] Let  $f\in W^{\bot}.$
	The mapping $P:V^{L}_{H}\stackrel{onto}{\longrightarrow}W$ is given by 
\begin{equation}\label{eq61}
	Pf=\sum_{k=1}^{m}\left\langle f|S^{-1}f_{k}\right\rangle f_{k}.
\end{equation}
From the Theorem  (\ref{T1} ), the frame operator $S:W\longrightarrow W$ is bijective.
Hence the range of $S^{-1}$ is  $W.$
That is,
$$~S^{-1}:W\longrightarrow W.~$$
So $S^{-1}f_{k}=g$ for some $g\in W.$
Hence  (\ref{eq61}) gives
\begin{eqnarray*}
Pf&=&\sum_{k=1}^{m}\left\langle f|S^{-1}f_{k}\right\rangle f_{k}\\
&=&\sum_{k=1}^{m}\left\langle f|g\right\rangle f_{k},~\mbox{~for some~}g\in W\\
&=&0~\mbox{~as~}g\in W\mbox{~and~}f\in W^{\bot}.
\end{eqnarray*}
Therefore 
\begin{equation}\label{eq62}
	Pf=0~\mbox{~~\text{for all}~}f\in W^{\bot}.
\end{equation}
\end{enumerate}
From (\ref{eq60}) and (\ref{eq62}), $P$ is an orthogonal projection.
\end{proof}

\begin{definition}\label{D6}
The numbers
\begin{equation}\label{eq63}
	\left\langle f|S^{-1}f_{k}\right\rangle,~k=1,\cdots,m
\end{equation}
are called \textit{frame coefficients.}
The frame $\left\{S^{-1}f_{k}\right\}_{k=1}^{m}$ is called the \textit{canonical dual} of $\left\{f_{k}\right\}_{k=1}^{m}.$
\end{definition}




\begin{thebibliography}{XXXX}
\bibitem{Ad}Adler, S.L., {\em Quaternionic quantum mechanics and quantum fields}, Oxford University Press, New York, 1995.

\bibitem{A1}Ali, S.T, Antoine, J-P, Gazeau, J-P., \textit{Continuous frames in Hilbert spaces}, Ann.Phys., {\bf 222} (1993), 1-37.

\bibitem{Alibk} Ali, S.T., Antoine, J-P, Gazeau, J-P, {\em Coherent states, wavelets, and their generalization},(2 nd edition) Springer-Verlag, New York, 2014.

\bibitem{Ole} Christensen, O., {\em An Introduction to Frames and Riesz Bases}, Birkhauser Boston, New York, 2003.

\bibitem{CG}Cotfas, N., Gazeau, J-P, {\em Finite tight frames and some applications}, J. Phys. A: Math. Theor. {\bf 43} (2010), 193001.

\bibitem{CGV} Cotfas, N., Gazeau, J-P., Vourdas, A., {\em Finite dimensional Hilbert spaces and frame quantization}, J. Phys. A: Math. Theor. {\bf 44} (2011), 17303.

\bibitem{D1}Daubechies, I., Grossmann, A., Meyer, Y., \textit{Painless nonorthogonal expansions}, J. Math. Phys. {\bf 27} (1986) 1271-1283. 

\bibitem{D2}Daubechies, I., \textit{Ten lectures on wavelets}, SIAM. Philadelphia, 1992.

\bibitem{DU} Duffin, R.J., Schaeffer, A.C.,
\textit{A class of nonharmonic Fourier series}, Trans. Amer. Math. Soc.. {\bf 72} (1952) 341-366.
    
\bibitem{GO} Goyal, V.K., Kovacevic, J., Kelner, A.J.,
\textit{Quantized frame expansions with erasures}, Appl. Comp. Harm. Anal. {\bf 10} (2000), 203-233.
 
\bibitem{GR} Grochenig, K.H.,
\textit{Foundations of time-frequency analysis}, Birkhauser, Boston, 2000.

\bibitem{ThiAli} Thirulogasanthar, K., Ali, S.T., {\em Regular subspace of a quaternionic Hilbert space from quaternionic Hermite polynomials and associated coherent states}, J. Math. Phys. {\bf 54} (2013), 013506.

\bibitem{TW}Thirulogasanthar, K., Bahsoun, W., \textit{Frames built on fractal sets}, J.Geom. Phys., {\bf 50} (2004) 79-98.

\bibitem{TH} Thirulogasanthar, K., Honnouvo, G., Krzyzak, A. \textit{Coherent states and Hermite polynomials on quaternionic Hilbert spaces}, J.Phys.A: Math. Theor., {\bf 43} (2010) 385205.

\bibitem {Role} Rolewicz, S., {\em Metric linear spaces}, D. Riedel publishing company, Holland, 1985.

\end{thebibliography}
\end{document}